\documentclass{sig-alternate-05-2015}

\usepackage{amsmath}
\usepackage{mathtools}
\usepackage{multirow}
\usepackage{dsfont}
\usepackage{rotating}
\usepackage{booktabs}
\usepackage{xspace}
\usepackage{mathtools, cuted}
\usepackage{arydshln} 
\usepackage{paralist} 
\usepackage{times}

\usepackage{algorithm}
\usepackage[noend]{algorithmic}
\makeatletter
\newcommand{\IFTHEN}[3][default]{\ALC@it\algorithmicif\ #2\
  \algorithmicthen\ #3\
  \ifthenelse{\boolean{ALC@noend}}{}{\algorithmicendif\ } \ALC@com{#1}}
\newcommand{\IFTHENELSE}[3]{\ALC@it\algorithmicif\ #1\
  \algorithmicthen\ #2\
  \algorithmicelse\ #3\
  \ifthenelse{\boolean{ALC@noend}}{}{\algorithmicendif\ } }
\newcommand{\IFTHENEND}[2]{\ALC@it\algorithmicif\ #1\
  \algorithmicthen\ #2\ \algorithmicendif}
\makeatother

\usepackage{color}

\makeatletter
\def\fnum@figure{Protocol \thefigure}
\makeatother

\usepackage{todonotes}

\newlength{\spaceafterprotocol}
\DeclareMathOperator{\rank}{rank}

\newtheorem{theorem}{Theorem}
\newtheorem{remark}{Remark}
\newtheorem{corollary}{Corollary}
\newtheorem{lemma}{Lemma}
\newtheorem{proposition}{Proposition}
\newcommand{\F}{\ensuremath{\mathbb{F}}}
\newcommand{\RPM}[1]{\ensuremath{\mathcal{R}_{#1}}\xspace}
\newcommand{\SO}[1]{\ensuremath{\widetilde{O}({#1})}\xspace}
\newcommand{\GO}[1]{\ensuremath{O ({#1})}\xspace}
\newcommand{\sample}[1]{\ensuremath{\xleftarrow{\$} {#1}}\xspace}

\usepackage{mathtools}

\usepackage{letltxmacro}
\LetLtxMacro\orgvdots\vdots
\LetLtxMacro\orgddots\ddots

\makeatletter
\DeclareRobustCommand\vdots{%
  \mathpalette\@vdots{}%
}
\newcommand*{\@vdots}[2]{%
  \sbox0{$#1\cdotp\cdotp\cdotp\m@th$}%
  \sbox2{$#1.\m@th$}%
  \vbox{%
    \dimen@=\wd0 %
    \advance\dimen@ -3\ht2 %
    \kern.5\dimen@
    \dimen@=\wd2 %
    \advance\dimen@ -\ht2 %
    \dimen2=\wd0 %
    \advance\dimen2 -\dimen@
    \vbox to \dimen2{%
      \offinterlineskip
      \copy2 \vfill\copy2 \vfill\copy2 %
    }%
  }%
}
\DeclareRobustCommand\ddots{%
  \mathinner{%
    \mathpalette\@ddots{}%
    \mkern\thinmuskip
  }%
}
\newcommand*{\@ddots}[2]{%
  \sbox0{$#1\cdotp\cdotp\cdotp\m@th$}%
  \sbox2{$#1.\m@th$}%
  \vbox{%
    \dimen@=\wd0 %
    \advance\dimen@ -3\ht2 %
    \kern.5\dimen@
    \dimen@=\wd2 %
    \advance\dimen@ -\ht2 %
    \dimen2=\wd0 %
    \advance\dimen2 -\dimen@
    \vbox to \dimen2{%
      \offinterlineskip
      \hbox{$#1\mathpunct{.}\m@th$}%
      \vfill
      \hbox{$#1\mathpunct{\kern\wd2}\mathpunct{.}\m@th$}%
      \vfill
      \hbox{$#1\mathpunct{\kern\wd2}\mathpunct{\kern\wd2}\mathpunct{.}\m@th$}%
    }%
  }%
}

\newenvironment{smatrix}{\left[\begin{smallmatrix}}{\end{smallmatrix}\right]}
\newcommand{\checks}[1]{\ensuremath{\mathrel{\stackrel{?}{#1}}}}
\newcommand{\minor}[3]{\ensuremath{\left [#1\right]^{\{#2\}}_{\{#3\}}}}

\usepackage{color,svgcolor}
\usepackage[plainpages=true]{hyperref}
\makeatletter
\hypersetup{
  pdftitle={\@title},
  pdfauthor={\@author},
  breaklinks=true,
  colorlinks=true,
  linkcolor=darkred,
  citecolor=blue,
  urlcolor=darkgreen,
}
\makeatother

\CopyrightYear{2017}
\setcopyright{licensedothergov}
\conferenceinfo{ISSAC '17}{July 25--28, 2017, Kaiserslautern, Germany}
\isbn{978-1-4503-5064-8/17/07}\acmPrice{\$15.00}
\doi{http://dx.doi.org/10.1145/3087604.3087609}

\newcommand{\OpenDreamKit}{the \href{http://opendreamkit.org}{OpenDreamKit} \href{https://ec.europa.eu/programmes/horizon2020/}{Horizon 2020} \href{https://ec.europa.eu/programmes/horizon2020/en/h2020-section/european-research-infrastructures-including-e-infrastructures}{European Research Infrastructures} project (\#\href{http://cordis.europa.eu/project/rcn/198334_en.html}{676541})}

\title{Certificates for triangular equivalence and rank profiles\texorpdfstring{\thanks{This work is partly funded by \OpenDreamKit.}}{}}
\begin{document}

\numberofauthors{3}

\def\sharedaffiliation{%
\end{tabular}
\begin{tabular}{c}}

\author{
\alignauthor
Jean-Guillaume Dumas\\
\alignauthor
David Lucas\\
\alignauthor
Cl\'ement Pernet\\
      \sharedaffiliation
    \affaddr{Universit\'e Grenoble Alpes, Laboratoire Jean Kuntzmann, CNRS, UMR 5224}\\
    \affaddr{700 avenue centrale, IMAG - CS 40700, 38058 Grenoble cedex 9, France}\\
    \email{\{firstname.lastname\}@univ-grenoble-alpes.fr}
}
 
\maketitle

\begin{abstract}
In this paper, we give novel certificates for triangular equivalence
and rank profiles. 
These certificates enable to verify the row or column rank profiles or the
whole rank profile matrix faster than
recomputing them, with a negligible overall overhead.
We first provide quadratic time and space non-interactive certificates
saving the logarithmic factors of previously known ones.
Then we propose interactive certificates for the same problems
whose Monte Carlo verification complexity requires a small constant
number of matrix-vector
multiplications, a linear space, and a linear number of extra field operations.
As an application we also give an interactive protocol, certifying the
determinant of dense matrices, faster than the best previously known one.
\end{abstract}

\section{Introduction}
Within  the setting of verifiable computing, we propose in this paper
{\em interactive certificates} with the taxonomy
of~\cite{dk14}.
Indeed, we consider a protocol where a {\em Prover} performs a
computation and provides additional data structures or exchanges with a
{\em Verifier} who will use these to 
check the validity of the result, faster than by just recomputing it.
%
More precisely, in an interactive certificate,
the Prover submits
a {\em Commitment}, that is some result of a computation;
the Verifier answers by
a {\em Challenge}, usually some uniformly sampled random values;
the Prover then answers with
a {\em Response}, that the Verifier can use to convince himself of the validity
of the commitment. 
Several {\em rounds} of challenge/response might be necessary
for the Verifier to be fully convinced.

By Prover (resp. Verifier) {\em time}, we thus mean bounds on the number of arithmetic
operations performed by the Prover (resp. Verifier) during the
protocol, while by extra {\em space}, we mean bounds on the volume of
data being exchanged, not counting the size of the input and output of the computation. 

Such protocols are said to be {\em complete} if the probability
that a true statement is rejected by the Verifier can be made arbitrarily
small; and {\em sound} if the probability that a false
statement is accepted by the Verifier can be made arbitrarily small.
In practice it is sufficient that those probabilities are $<1$, as the protocols
can always be run several times.
Some certificates will also be {\em perfectly complete}, that
is a true statement is never rejected by the Verifier. 
%
All these certificates can be simulated non-interactively by Fiat-Shamir
heuristic~\cite{Fiat:1986:Shamir}: uniformly sampled random values
produced by the Verifier are replaced by cryptographic hashes of the
input and of previous messages in the protocol. Complexities are
preserved.

We do not use generic approaches to verified computation
(where protocols check circuits with polylogarithmic
depth~\cite{Goldwasser:2008:delegating} or use amortized
models and homomorphic
encryption~\cite{Costello:2015:gepetto}).
Rather, we use dedicated certificates
as those designed for
dense~\cite{freivalds79,kns11} or
sparse~\cite{dk14,jgd:2016:gammadet} exact linear
algebra.
The obtained certificates are problem-specific,
but try to reduce as much as possible the
overhead for the Prover, while preserving a fast verification procedure.

We will consider an $m\times n$  matrix $A$ of rank $r$ over a field $\F$.
%
The \emph{row rank profile}
of $A$ is the lexicographically minimal sequence of $r$ indices of independent
rows
of $A$.
Matrix $A$ has \emph{generic row
  rank profile} if its row
rank profile is $(1, \dots, r)$.
The \emph{column rank profile} is defined similarly on the columns of $A$.
Matrix $A$ has generic rank profile if its $r$ first leading principal minors
are nonzero. 
The \textit{rank profile matrix} of $A$, denoted by $\mathcal{R}_A$ is the
unique $m \times n$ $\{0,1\}$-matrix with $r$ nonzero entries, of which every
leading sub-matrix has the same rank 
as the corresponding sub-matrix of $A$. It is possible to compute
$\mathcal{R}_A$ with a deterministic algorithm in
$\GO{mnr^{\omega-2}}$ or with a Monte-Carlo probabilistic algorithm in $(r^{\omega} + m + n + \mu(A))^{1 + o(1)}$ field operations \cite{dps16},
where $\mu(A)$ is the arithmetic cost to multiply $A$ by a vector.

We first propose quadratic, space and verification time,
  non-in\-ter\-ac\-tive practical certificates for the row or column rank profile and
  for the rank profile matrix that are rank-sensitive. Previously
  known certificates have additional logarithmic factors to the
  qua\-dra\-tic complexities: replacing matrix multiplications by
  qua\-dra\-tic verifications in recursive algorithms yields at least one $\log(n)$
  factor~\cite{kns11}, graph-based approaches cumulate this and other
  logarithmic factors, at least from a compression by magical graphs
  and from a dichotomic search~\cite{sy15}.  

  We then propose two  linear space interactive
  certificates: one certifying that two non-singular matrices are
  triangular equivalent, i.e. there is a triangular change of basis from one to
  the other; the other one, certifying that a matrix has a generic rank profile.
  These certificates are then applied to certify the row or column rank profile,
  the $Q$ (permutation) and $D$ (diagonal) factors of a LDUP factorization, the
  determinant and the rank profile matrix.
  These certificates require, for the Verifier, between 1 and 3 applications of
  $A$ to a vector and a linear amount of field operations. They are still
  elimination-based for the Prover, but do not require to 
  communicate the obtained triangular decomposition. 
%
For the Determinant, this new certificates require the computation of a PLUQ
decomposition for the Prover, linear communication and Verifier time, with no
restriction on the field size.

Table~\ref{tab:io} compares linear quadratic volumes of communication, as
well as sub-cubic (PLUQ, {\sc{CharPoly}}) or quadratic matrix operations (one
matrix-vector multiplication with a dense matrix is denoted \texttt{fgemv}). 
The results shows first that it is interesting to use linear space certificates
even when they have quadratic Verification time. 
The table also presents a practical constant
factor of about 5 between PLUQ and {\sc{CharPoly}} computations. 
Computations use the
FFLAS-FFPACK library (\url{http://linbox-team.github.io/fflas-ffpack})
on a single Intel Skylake core @3.4GHz, while we measured some communications
between two workstations over an Ethernet Cat. 6, @1Gb/s network cable.

\begin{table}[htbp]
  \centering
  \begin{tabular}{lrrr}
    \toprule
    Dimension		& $2k$& $10k$&$50k$ \\
    \midrule
    PLUQ		& 0.28s	& 17.99s	& 1448.16s\\
    {\sc{CharPoly}}	& 1.96s & 100.37s & 8047.56s\\
    \midrule
    Linear comm.	& 0.50s	& 0.50s	& 0.50s\\
    Quadratic comm.	& 1.50s	& 7.50s	& 222.68s\\
    \midrule
    \texttt{fgemv}	& 0.0013s& 0.038s& 1.03s\\
    \bottomrule
  \end{tabular}
  \caption{Communication of 64 bit words versus computation modulo
    $131071$}\label{tab:io}
\end{table}

A summary of our contributions is given in
Table~\ref{contributions}, to be compared with the state of the art
in Table~\ref{SotA_matrix_rank}.
\begin{table*}[htbp]\small
    \centering
    \begin{tabular}{llllllll}
        \toprule
         & Algorithm & Inter.  &  \multicolumn{2}{c}{Prover} & \multirow{2}{*}{Communication} &
        Probabilistic & \multirow{2}{*}{$\#\F$} \\
        \cmidrule{4-5}
        & & & Determ.& Time & & Verifier Time\\
        \midrule
        \multirow{3}{*}{\sc{Rank}} &\cite{kns11} over \cite{ckl13} & No & No &$\SO{r^{\omega} + \mu(A)}$ & $\SO{r^2+m+n}$ & $\SO{r^{2} + \mu(A)}$&$\geq 2$ \\
        &\cite{dk14} & Yes & No & $O(n(\mu(A)+n))$ & $O(m+n)$ & $2\mu(A) + \SO{m + n}$& $\SO{min\{m,n\}}$ \\
        &\cite{eberly15} & Yes & Yes & $O(mnr^{\omega -2})$ & $O(m+r)$ & $O(r + \mu(A) + m + n)$ &$\geq 2$\\
        \midrule
        \multirow{2}{*}{CRP/RRP} & \cite{kns11} over \cite{sy15} & No & No
        &$\SO{r^{\omega} +m+n+ \mu(A)}$ & $\SO{r^2+m+n}$ & $\SO{r^2+ m + n +
          \mu(A)}$ & $\SO{min\{m,n\}}$\\
         & \cite{kns11} over \cite{jps13} & No & Yes &$O(mnr^{\omega-2})$ & $\SO{mn}$ & $\SO{mn}$ &$\geq 2$\\
        \midrule
       \multirow{2}{*}{RPM} & \cite{kns11} over \cite{dps16} & No & No &$\SO{r^{\omega} + m + n + \mu(A)}$ & $\SO{r^2+m+n}$ & $\SO{r^2 + m + n + \mu(A)}$ & $\SO{min\{m,n\}}$\\
        & \cite{kns11} over \cite{DPS:2013} & No & Yes &$O(mnr^{\omega-2})$ & $\SO{mn}$ & $\SO{mn}$ &$\geq 2$\\
       \midrule
            \multirow{2}{*}{\textsc{Det}}& \cite{freivalds79} \& PLUQ  & No & Yes & $O(n^\omega)$ & $O(n^2)$ & $O(n^2) + \mu(A)$ &$\geq 2$ \\
            & \cite{jgd:2016:gammadet} \& {\sc{CharPoly}}   &
        Yes &No &
        $O(n\mu(A))$ or $O(n^\omega)$  & $O(n)$ & $\mu(A)+O(n)$  &$\geq n^2$\\
        \bottomrule
    \end{tabular}
    \caption{State of the art certificates for the rank, the row and column rank
      profiles, the rank profile matrix and the determinant}
    \label{SotA_matrix_rank}
\end{table*}
\begin{table*}[htbp]\small
    \centering
    \begin{tabular}{llllllll}
        \toprule
        &Algorithm & Interactive  & \multicolumn{2}{c}{Prover} & \multirow{2}{*}{Communication} &
        Probabilistic & \multirow{2}{*}{$\#\F$}\\
        \cmidrule{4-5}
        & & & Deterministic & Time & & Verifier Time\\
        \midrule
        \multirow{2}{*}{CRP/RRP} & \S~\ref{sec:noninterractive:CRP}& No & Yes
        &$O(mnr^{\omega-2})$ & $O(r(m+n))$ & $O(r(m+n))+ \mu(A)$ &$\geq 2$\\
        & \S~\ref{sec:interractive:CRP} & Yes & Yes & $O(mnr^{\omega-2})$ & $O(m+n)$ & $2\mu(A)+O(m+n)$ &$\geq 2$ \\
        \midrule
        \multirow{2}{*}{RPM} & \S~\ref{sec:noninterractive:RPM} & No & Yes &$O(mnr^{\omega-2})$  & $O(r(m+n))$ & $O(r(m+n)) + \mu(A)$ &$\geq 2$ \\
        & \S~\ref{sec:interractive:RPM}& Yes & Yes & $O(mnr^{\omega-2})$ & $O(m+n)$ & $4\mu(A)+O(m+n)$ &$\geq 4$ \\
        \midrule
        {\sc{Det}}& \S~\ref{sec:det} \& PLUQ & Yes & Yes & $O(n^\omega)$ & $O(n)$ & $\mu(A)+O(n)$  &$\geq 2$\\
        \bottomrule
    \end{tabular}
    \caption{This paper's contributions}
    \label{contributions}
\end{table*}
%
%
We identify the symmetric group with the group of permutation matrices, and
write $P\in \mathcal{S}_n$ to denote that a matrix $P$ is a permutation
matrix. There, $P[i]$ is the row index of the nonzero element 
of its $i$-th column; $\mathcal{D}_n(\F)$ is the group of
invertible diagonal matrices over the field $\F$ and $[A]^I_J$ is the
$(I,J)$-minor of the matrix $A$ (the determinant of the submatrix of
$A$ with row indices in $I$ and column indices in~$J$).
Lastly, $x \sample{S}$ denotes that $x$  is sampled uniformly at random
from~$S$.
\section{Non interactive and quadratic communication certificates}
\label{sec:noninterractive}
In this section, we propose two certificates, first for the column (resp. row)
rank profile, and, second, for the rank profile matrix.
While the certificates have a quadratic space communication complexity, they have the advantage
of being non-interactive.

\subsection{Freivalds' certificate for matrix product}
\label{subsec:freivalds}

In this paper, we will use Freivalds' certificate \cite{freivalds79} to verify matrix multiplication.
Considering three matrices $A, B$ and $C$ in $\mathbb{F}^{n \times n}$, such that 
$A \times B = C$, a straightforward way of verifying the equality would be to perform the multiplication
$A \times B$ and to compare its result coefficient by coefficient with $C$. While this method is
deterministic, it has a time complexity of $O(n^\omega)$, which is the matrix multiplication complexity.
As such, it cannot be a certificate, as there is no complexity difference between the computation and
the verification.

\begin{figure}[htbp]
    \centering
    \begin{tabular}{|l c l|}
        \hline
        Prover & & Verifier \\
        & $A, B\in \mathbb{F}^{n \times n}$
        & \\
        \hdashline\rule{0pt}{12pt}
        $C= A  B $ & $\xrightarrow{\text{C}}$ &
        $v \in \mathbb{F}^{n \times 1}$\\
        & & $ A  (B  v)- C  v \stackrel{?}{=} 0 $\\\hline
\end{tabular}
    \caption{Freivalds' certificate for matrix product}
    \label{cert:freivalds}
    \vspace{\spaceafterprotocol}
\end{figure}

Freivalds' certificate proposes a probabilistic method to check this product in a time complexity of
$\mu(A)+\mu(B)+\mu(C)$ using matrix/vector multiplication, as detailed in
Figure~\ref{cert:freivalds}.

\subsection{Column rank profile certificate}
\label{sec:noninterractive:CRP}

We now propose a certificate for the column rank profile.

\begin{figure}[htbp]
  \centering
  \begin{tabular}{|p{75pt} c p{90pt}|}
    \hline
    \multicolumn{1}{|c}{Prover} & & \multicolumn{1}{c|}{Verifier} \\
    & $A \in \mathbb{F}^{m \times n}$ & \\
    \hdashline
    \multirow{2}{80pt}{A $PLUQ$ decomposition of $A$ s.t. $UQ$ is in
      row echelon form} 
    &\multirow{2}{*}{$\xrightarrow{\text{P,L,U,Q}}$} 
    &\rule{0pt}{10pt}{$UQ$ row echelonized?}\\
    \rule{0pt}{30pt}&  & $A\stackrel{?}{=}PLUQ$, by cert.~\ref{cert:freivalds}\\
    \hdashline\rule{0pt}{10pt}
    & & Return $Q[1],\ldots, Q[r]$\\
    \hline
  \end{tabular}
  \caption{Column rank profile, non-interactive}
  \label{cert:crp:noninter}
\vspace{\spaceafterprotocol}\end{figure}

\begin{lemma}
    \label{lem:crp:ni}
    Let $A=PLUQ$ be the PLUQ decomposition of an $m \times n$ matrix $A$ of rank
    $r$. If  $UQ$ is in row echelon form then $(Q[1], \dots, Q[r])$ is the column rank profile of $A$.
\end{lemma}

\begin{proof}
  Write $A=P\begin{smatrix} L_1 \\ L_2 \end{smatrix}\begin{smatrix} U_1 & U_2 \end{smatrix}Q$, where $L_1$ and
  $U_1$ are $r\times r$ lower and upper triangular respectively.
  If $UQ$ is in echelon form, then  
  $ R  = \begin{smatrix} \begin{smallmatrix}I_r & U_{1}^{-1} U_2\end{smallmatrix} \\ {0_{(m-r)\times n}} \end{smatrix}$
  is in reduced echelon form.
  Now
  $$
  \begin{bmatrix} U_1^{-1}\\&I_{m-r} \end{bmatrix}
  \begin{bmatrix} L_1\\L_2&I_{m-r}  \end{bmatrix}^{-1} P^TA= 
  \begin{bmatrix} U_1^{-1}UQ\\0_{(m-r)\times n} \end{bmatrix}=
  R$$
  is left equivalent to $A$ and is therefore the echelon form of $A$.
  Hence the sequence of column positions of the pivots in $R$, that is
  $(Q[1],\dots,Q[r])$, is the 
  column rank profile of $A$.
\end{proof}

Lemma~\ref{lem:crp:ni} provides a criterion to verify a column rank
profile from a PLUQ decomposition.
Such decompositions can be computed in practice by several variants of Gaussian
elimination, with no arithmetic overhead, as shown in~\cite{jps13} or~\cite[\S~8]{dps15}.
Hence, we propose the certificate in Protocol~\ref{cert:crp:noninter}.

\begin{theorem}
    Let $A \in \mathbb{F}^{m \times n}$ with $r=\rank(A)$.
    Certificate~\ref{cert:crp:noninter},  verifying the column rank profile of $A$
    is sound, perfectly complete, with a communication bounded by $O(r(m+n))$, a Prover
    computation bounded by $O(mnr^{\omega -2})$ and
    a Verifier computation cost bounded by $O(r(m+n)) + \mu(A)$.
\end{theorem}

\begin{proof}
    If the Prover is honest, then, $UQ$ will be in row echelon form and $A=PLUQ$, thus, 
    by Lemma~\ref{lem:crp:ni}, the Verifier will be able to read the column rank
    profile of~$A$ from~$Q$.
    If the Prover is dishonest, either $A \neq PLUQ$, which will be caught by the Prover with probabilty
    $p \geq 1-\frac{1}{q}$ using Freivalds' certificate \cite{freivalds79} or $UQ$ is not in row echelon from, which
    will be caught every time by the Verifier.

    The Prover sends $P, L, U \text{ and } Q$ to the Verifier, hence the communication cost of $O(r(m+n))$, as
    $P$ and $Q$ are permutation matrices and $L, U$, are respectively $m \times r$ and $r \times n$ matrices,
    with $r = rank(A)$.
    Using algorithms provided in \cite{jps13}, one can compute the expected $PLUQ$ decomposition in
    $O(mnr^{\omega -2})$.
    The Verifier has to check if $A = PLUQ$, and if $UQ$ is in row echelon form, which can be done in $O(r(m+n))$.
\end{proof}

Note that this holds for the row rank profile of $A$: in that case, the Verifier has to check if $PL$ is in
column echelon form.

\subsection{Rank profile matrix certificate}
\label{sec:noninterractive:RPM}



\begin{lemma}\label{lem:echelonized}
  A decomposition  $A=PLUQ$
  reveals the rank profile matrix, namely
  $\RPM{A}=P \begin{smatrix}  I_r\\&0  \end{smatrix}Q$, if and only if
  $P \begin{smatrix} L&0  \end{smatrix}P^T$ is lower triangular and
  $Q^T \begin{smatrix} U\\0  \end{smatrix} Q$ is upper triangular.
\end{lemma}

\begin{proof}

  The \textit{only if} case is proven in~\cite[Th.~21]{dps16}.
Now suppose that 
$P\begin{smatrix}L&0_{m\times(m-r)} \end{smatrix}P^T$ is lower triangular.
Then we must also have that
$\overline{L}=P\begin{smatrix}L&
\begin{smallmatrix}  0\\I_{m-r}\end{smallmatrix}
 \end{smatrix}P^T$
is lower triangular and non-singular.  Similarly suppose that
$Q^T \begin{smatrix} U\\0  \end{smatrix} Q$ is upper triangular so that
$\overline{U}=Q^T\begin{smatrix}U\\
\begin{smallmatrix}  0&I_{n-r}\end{smallmatrix}\end{smatrix}
Q$
is non-singular upper triangular.
We have $A=\overline{L}P \begin{smatrix}  I_r\\&0\end{smatrix}Q\overline{U}$.
Hence the rank of any $(i,j)$ leading submatrix of $A$ is that of the $(i,j)$
leading submatrix of $P\begin{smatrix}  I_r\\&0\end{smatrix}Q$, thus proving
that $\RPM{A}=P\begin{smatrix}  I_r\\&0\end{smatrix}Q$.
\end{proof}

We use this characterization to verify the computation of the rank
profile matrix in the following protocol:
Once the Verifier receives $P, L, U \text{ and } Q$, he has to:
\begin{compactenum}
\item Check that $A = PLUQ$, using Freivalds' certificate \cite{freivalds79} 
\item Check that $L$ is echelonized by $P$ and $U^T$ by~$Q^T$.
\item If successful, compute the rank profile matrix of $A$ as
  $\mathcal{R}_A = P \begin{bsmallmatrix} I_{r} & \\ & 0_{(m-r) \times (n-r)} \end{bsmallmatrix} Q$
\end{compactenum}

\begin{figure}[htbp]
    \centering
    \begin{tabular}{|c c p{4cm}|}
      \hline
       Prover & & \hspace{15pt}Verifier \\
        \multicolumn{3}{|c|}{\hspace{-45pt}$A \in \mathbb{F}^{m \times n}$} \\
        \hdashline\rule{0pt}{12pt}
         \multirow{3}{2.4cm}{a PLUQ decomp. of $A$ revealing  $\RPM{A}$.} &
         $\xrightarrow{\text{P,L,U,Q}}$ &
         1.  $A\stackrel{?}{=}PLUQ$ by
         Protoc.~\ref{subsec:freivalds}\\
         & & 2. Is $PLP^T$ lower triangular?\\
         & & 3. Is $Q^TUQ$ upper triangular?\\
        \hline
    \end{tabular}
    \caption{Rank profile matrix, non-interactive}
    \label{cert:RPM:noninter}
\vspace{\spaceafterprotocol}\end{figure}

\begin{theorem}
    Certificate~\ref{cert:RPM:noninter} verifies the rank profile matrix of $A$,
    it is sound and perfectly complete, with a communication cost bounded by
    $O(r(n+m))$, a Prover computation cost bounded by $O(mnr^{\omega-2})$ and a
    Verifier computation cost bounded by $O(r(m+n)) + \mu(A)$. 
\end{theorem}

\begin{proof}
    If the Prover is honest, then, the provided $PLUQ$ decomposition is indeed
    a factorization of $A$, which means Freivalds' certificate will pass.
    It also means this $PLUQ$ decomposition reveals the rank profile matrix.
    According to Lemma~\ref{lem:echelonized}, $PLP^T$ will be lower triangular
    and $Q^TUQ$ upper triangular. Hence the verification will succeeds and 
    $\mathcal{R}_A = P \begin{bsmallmatrix} I_{r} & \\ & 0_{(m-r) \times (n-r)} \end{bsmallmatrix} Q$
    is indeed the rank profile matrix of $A$.
    If the Prover is dishonest, either $A \neq PLUQ$, which will be caught with probabilty
    $p \geq 1-\frac{1}{q}$ by Freivalds' certificate or the $PLUQ$ decomposition does not reveal the 
    rank profile matrix of $A$.
    In that case, Lemma~\ref{lem:echelonized} implies that either
    $P \begin{smatrix}      L&0    \end{smatrix}P^T$ is not lower triangular or
    $P \begin{smatrix}      U\\ 0 \end{smatrix}Q$ is not upper triangular which
    the will be detected.

    The Prover sends $P, L, U \text{ and } Q$ to the Verifier, hence the
    communication cost of $\GO{(n+m)r}$.
    A rank profile matrix revealing $PLUQ$ decomposition can be computed in $O(mnr^{\omega-2})$ operations~\cite{DPS:2013}.
    The Verifier has to check if $A = PLUQ$, which can be achieved in
    $O((m+n)r)+\mu(A)$ field operations.
\end{proof}

\section{Linear communication certificate toolbox}

\subsection{Triangular one sided equivalence}

Two matrices $A, B \in \mathbb{F}^{m \times n}$ are right (resp. left)
equivalent if there exist an invertible $n\times n$ matrix $T$ such that $AT=B$
(resp. $TA=B$). If in addition $T$ is a lower triangular matrix, we say that $A$ and
$B$ are lower triangular right (resp. left)  equivalent. The upper triangular
right (resp. left ) equivalence is defined similarly.
We propose a certification protocol that two matrices are left or right
triangular equivalent.
Here, $A$ and $B$ are input, known by the Verifier and the Prover.
A simple certificate would be the matrix $T$ itself, in which case the Verifier would
check the product $AT = B$ using Freivalds' certificate.
This certificate is non-interactive and requires a quadratic amount of communication.
In what follows, we present a certificate which allows to verify the one sided
triangular equivalence without communicating $T$, requiring only $2n$
communications.
It is essentially a Freivalds' certificate with a more constrained interaction
pattern in the way the challenge vector and the response vector are
communicated.
This pattern imposes a triangular structure in the way the Provers' responses
depend on the Verifier challenges which match with the structure of the problem.

\begin{figure}[htbp]
    \centering
    \begin{tabular}{|p{2.7cm} c l|}
        \hline
        Prover &  & Verifier \\
        & $ A, B \in \mathbb{F}^{m{\times}n} $ &  \\
        & $A$ regular, $m{\geq}n$ &  \\
        \hdashline\rule{0pt}{12pt}
        $T$ lower triangular matrix s.t. $AT = B$ & $\xrightarrow{\text{1: T exists}}$ & \\
        \hdashline\rule{0pt}{12pt}
        \multirow{2}{*}{$y_1 = T_{1,*} \begin{smatrix}x_1\\0\\ \vdots \end{smatrix} $}& $\xleftarrow{\text{$2: x_1$}}$ & $x_i\sample{S}\subset{\F}$ \\
         & $\xrightarrow{\text{$3: y_1$}}$ & \\
        \multicolumn{1}{|c}{$\vdots$}&$\vdots$& \\
        \multirow{2}{*}{$y_n = T_{n,*} \begin{smatrix}x_1\\ \vdots\\x_n \end{smatrix} $}& $\xleftarrow{\text{$2n: x_n$}}$ &  \\
        & $\xrightarrow{\text{$2n+1: y_n$}}$ &
        $y =\begin{bmatrix}y_1&..&y_n\end{bmatrix}^T$ \\
        & &
        $Ay\stackrel{?}{=}Bx$\\
        \hline
    \end{tabular}
    \caption{Lower triang. right equivalence of regular matrices}
    \label{cert:triangular_equivalence}
\vspace{\spaceafterprotocol}\end{figure}

\begin{theorem}
    \label{th:triangular_equivalence}
    Let $A, B \in \mathbb{F}^{m \times n}$, and assume $A$ is regular.
    Certificate~\ref{cert:triangular_equivalence}
    proves that there exists a lower triangular matrix $T$ such that $AT = B$.
    This certificate is sound, with probabilty larger than $1-\frac{1}{|S|}$,
    perfectly complete, occupies $2n$ 
    communication space, and can be computed in 
    $O(mn^{\omega-1})$ field operations and  
    verified in $\mu(A) + \mu(B)$ field operations.
\end{theorem}

\begin{proof}
  If the Prover is honest, then $AT=B$ and she just computes $y=Tx$, so
  that $Ay=ATx=Bx$.
  If the Prover is dishonest, 
  replace the random values $x_1,\dots,x_n$ by algebraically independent variables
  $X_1,\dots,X_n$.
  Since $A$ is regular, there is a unique $n\times n$ matrix
  $T$ (that is, $T=A^{\dagger}B$ with $A^{\dagger}$ the Moore-Penrose inverse
  of $A$) such that $AT=B$. 
  For the same reason, there is a unique vector $\widehat{Y}=A^{\dagger}BX$ such
  that $A\widehat{Y}=BX$. 
  The vector $\widehat{Y}$ is then formed by $n$ degree-$1$
   polynomials in   $X_1,\dots,X_n$.
  If $T$ is not lower triangular, let $i$ be the first row such that
  $T_{i,j}\neq 0$ for some $j>i$, and let $j_m$ be the largest such $j$.
  Then $\widehat{Y}_i$ has degree 1 in $X_{j_m}$.
  Let $Y$ be the vector output by the Prover.
  At step $2i+1$, the value for $X_{j_m}$ was still not released, hence $Y_i$ is
  constant in $X_{j_m}$.
  As $A$ is regular, the verification $AY=BX=A\widehat{Y}$ is equivalent to
  $Y-\widehat{Y}=0$. The $i$-th component in this equation is $Y_i-\widehat{Y}_i=0$, whose
  left hand-side contains a non zero monomial in $X_j$. There is therefore a
  probability lower than $1/|S|$ that the random choice for $x_j$ makes this
  polynomial vanish.

  This certificate requires to transmit $x$ and $y$, which costs $2n$ in communication.
    The Verifier has to compute $Ay$ and $Bx$, whose computational cost is $\mu(A) + \mu(B)$.
    The Prover has to compute $T$, this can be done by a PLUQ
    elimination on $A$ followed by a triangular system solve, both in
    $O(mn^{\omega-1})$. Then  $y=Tx$ requires only $O(n^2)$ operations.
\end{proof}

    Note that the case where $T$ is upper triangular works similarly: the Verifier needs to transmit
    $x$ in reverse order, starting by $x_n$.



\subsection{Generic rank profile-ness}

The problem here is to verify whether a non-singular input matrix
$A\in\F^{m\times n}$ has generic rank profile (to test non-singularity, one can
apply beforehand the linear communication certificate in~\cite[Fig.~2]{dk14},
see also Protocol~\ref{cert:lower_rank} thereafter).
A matrix $A$ has generic rank profile if and only if it has an LU decomposition
$A=LU$, with $L$ unit lower triangular and $U$ non-singular upper
triangular.  The protocol picks random vectors $\phi,\psi,\lambda$ and asks the Prover to
provide the vectors $z^T=\lambda^TL$, $x=U\phi$, $y=U\psi$ on the fly, while receiving the
coefficients of the vectors $\phi,\psi,\lambda$ one at a time.
These vectors satisfy the fundamental equations  
$z^Tx=\lambda^TA\phi$ and $z^Ty=\lambda^TA\psi$ that will be checked by the Verifier.





\begin{figure}[htbp]
  \centering
  \begin{tabular}{|l c l|}
    \hline\rule{0pt}{12pt}   
    Prover & & Verifier \\
     & $A \in \mathbb{F}^{n{\times}n}$ & \\
     & non-singular & \\
    \hdashline\rule{0pt}{12pt}    
     $A=LU $ &$\xrightarrow{\text{A has g.r.p.}}$ &\\
    \hdashline\rule{0pt}{12pt}   
    &\multicolumn{2}{l|}{\textbf{for} $i$ from $n$ downto $1$} \\
    $\begin{bmatrix} x&y \end{bmatrix} = U \begin{bmatrix}\phi&\psi \end{bmatrix}$ & $\xleftarrow{\phi_i,\psi_i}$ & \hspace{-2pt}$(\phi_i,\psi_i) \sample{S^2}\subset{\F^2}$\\
    & $\xrightarrow{x_i,y_i} $ &\\
    $z^T = \lambda^T L$& $\xleftarrow{\lambda_i}$ & \hspace{-2pt}$\lambda_i\sample{S}\subset{\F}$\\
    & $\xrightarrow{z_i} $ &\\
    \hdashline\rule{0pt}{12pt}    
    && \hspace{-2pt}$z^T  \begin{bmatrix}x&y\end{bmatrix}  \checks{=} (\lambda^T A)\begin{bmatrix}\phi&\psi\end{bmatrix}$ \\
    \hline
  \end{tabular}
  \caption{Generic rank profile with linear communication}
  \label{cert:GRP}
      \vspace{\spaceafterprotocol}
\end{figure}

\begin{theorem}
\label{th:GRP}
   Certificate~\ref{cert:GRP} verifying that a non-singular matrix has generic
   rank profile is sound, with probability larger than $1-\frac{1}{|S|}$,
   perfectly complete, communicates $3n$ field elements, and can be computed in
   $O(n^\omega)$ field operations for the Prover and $\mu(A) + 8n$ field
   operations for the Verifier. 
\end{theorem}

We will need the following Lemma, used in Dodgson determinant condensation rule.
\begin{lemma}[Desnanot-Jacobi, or Dodgson rule~\cite{Dod1866}]\label{lem:dodgson}
  \[
\minor{A}{1..n}{1..n} \minor{A}{2..n-1}{2..n-1} =  
\left|
\begin{array}{cc}
  \minor{A}{1..n-1}{1..n-1} &    \minor{A}{2..n}{1..n-1} \\
\minor{A}{1..n-1}{2..n} &  \minor{A}{2..n}{2..n} \\  
\end{array}
  \right|.
\]
\end{lemma}

Applying the same permutation, the cyclic shift of order 1 to the left, on the rows and columns of $A$, yields the
following formula with no change of sign:
  \begin{equation}\label{eq:dodgson:variant}
\minor{A}{1..n}{1..n} \minor{A}{1..n-2}{1..n-2} = 
\left|\begin{array}{cc}
\minor{A}{1..n-2, n}{1..n-2, n} &\minor{A}{1..n-1}{1..n-2, n} \\
  \minor{A}{1..n-2, n}{1..n-1}&\minor{A}{1..n-1}{1..n-1}
\end{array}
\right|.
  \end{equation}
    
\begin{proof}[of Theorem~\ref{th:GRP}]
The protocol is perfectly complete: if $A=LU$, then
$z^T\begin{bmatrix}x&y\end{bmatrix}=
\lambda^TLU\begin{bmatrix}\phi&\psi\end{bmatrix}=\lambda^TA\begin{bmatrix}\phi&\psi\end{bmatrix}$. 

Now, for the soundness, replace every $\phi,\psi,\lambda$ chosen at random by the Verifier by
vectors of algebraically independent variables $\Phi,\Psi,\Lambda$.
Similarly, the responses of the Prover $z,x,y$ are now 
vectors of algebraically independent variables $Z,X,Y$.
Under the assumption of  the success of the Verifier test,
\begin{equation}\label{eq:verif}
\left\{
\begin{array}{lll}
  Z^TX & = & \Lambda^T A \Phi\\
  Z^TY & = & \Lambda^T A \Psi\\
\end{array}
\right. ,
\end{equation}
and that $A$ is non-singular,
we will prove the following induction hypothesis:
\[
\text{H}_i: 
\left\{
  \begin{array}{lll}
 Z_{i\dots n}^T X_{i\dots n}  &=& \frac{1}{d_{i-1}}\sum_{i\leq j,k\leq n} \Lambda_k \minor{A}{1\dots i-1,k}{1\dots i-1,j}\Phi_j\\
 Z_{i\dots n}^T Y_{i\dots n}  &=&  \frac{1}{d_{i-1}}\sum_{i\leq j,k\leq n} \Lambda_k \minor{A}{1\dots i-1,k}{1\dots i-1,j}\Psi_j \\
    d_j\neq 0\ \forall j< i
  \end{array}
  \right.
  \]
where
$ d_i= \minor{A}{1\dots i}{1\dots i}$, $d_0=1$.

  For $i=1$, note that $ \minor{A}{1\dots i-1,k}{1\dots i-1,j} =A_{k,j}$, hence
  the right handsides of the first two equations of $H_1$ can be written as:
\[
\begin{array}{l}
  \sum_{1\leq j,k\leq n} \Lambda_k A_{k,j}\Phi_j =     \Lambda^TA\Phi = Z^TX\\
\sum_{1\leq j,k\leq n} \Lambda_k A_{k,j}  \Psi_j =     \Lambda^TA\Psi = Z^TY
\end{array}
  \]
by~\eqref{eq:verif}. Finally $d_0=1$ is obviously nonzero.

  Now suppose $H_{i}$ is true for some $0\leq i < n$. Then
  \begin{equation}\label{eq:sys:rec}
   \left\{
  \begin{array}{r}
    Z_iX_i + Z_{i+1..n}^T X_{i+1..n}  = \frac{1}{d_{i-1}} \Lambda_{i} \sum_{j=i}^n
    \minor{A}{1..i}{1..i-1,j} \Phi_j \\  +   \frac{1}{d_{i-1}} \sum_{j=i}^n\sum_{k=i+1}^n \Lambda_k
    \minor{A}{1..i-1,k}{1..i-1,j}\Phi_j \\
    Z_iY_i + Z_{i+1..n}^T Y_{i+1..n}  =  \frac{1}{d_{i-1}}\Lambda_{i}
    \sum_{j=i}^n \minor{A}{1..i}{1..i-1,j} \Psi_j \\
     + \frac{1}{d_{i-1}}  \sum_{j=i}^n\sum_{k=i+1}^n \Lambda_k \minor{A}{1..i-1,k}{1..i-1,j}\Psi_j
  \end{array}
  \right. .
  \end{equation}

At the time of choosing the value for $\Lambda_{i}$, all variables are set, except
$Z_i$.
Hence
for all value assigned to $\Lambda_i$, there is a value for $Z_i$ that
satisfies the above system of two linear equations in $Z_i$ and $\Lambda_i$.
Consequently this system is singular and the following two determinants vanish:
  \begin{equation}\label{eq:det1}
\left|
\begin{array}{cc}
  d_{i-1}X_i & \sum_{j=i}^n \minor{A}{1..i}{1..i-1,j} \Phi_j\\
  d_{i-1}Y_i & \sum_{j=i}^n    \minor{A}{1..i}{1..i-1,j} \Psi_j
\end{array}
\right| =0 
\end{equation}
\begin{equation}\label{eq:det2}
\left|
\begin{array}{cc}
\displaystyle\sum_{j=i}^n \minor{A}{1..i}{1..i-1,j} \Phi_j &  d_{i-1}Z_{i+1..n}^T X_{i+1..n}-F_A(\Lambda,i,\Phi)\\
\displaystyle\sum_{j=i}^n \minor{A}{1..i}{1..i-1,j} \Psi_j   &  d_{i-1}Z_{i+1..n}^T Y_{i+1..n}-F_A(\Lambda,i,\Psi)
\end{array}
\right| =0
\end{equation}
where $F_A(\Lambda,i,R)=\sum_{j=i}^n\sum_{k=i+1}^n
\Lambda_k\minor{A}{1..i-1,k}{1..i-1,j}R_j$, for $R=\Phi,\Psi$.
Actually, Equation~\eqref{eq:det2} is thus of the form $\left|
  \begin{array}{cc}
    d_i \Phi_i + b & a\Phi_i + e\\
    d_i \Psi_i + c & a\Psi_i + f
  \end{array}
  \right|=0$
  where $d_i=\minor{A}{1..i}{1..i}$,
  $a=-\sum_{k=i+1}^n\Lambda_k\minor{A}{1..i-1,k}{1..i}$
  and $b,c,e,f$ are constants with respect to the variables $\Phi_i,\Psi_i$.
  
  If $d_i=0$, then, at least one $\minor{A}{1..i}{1..i-1,j}$ for $j>i$ must be
  nonzero, otherwise $A$ would be singular.
  Similarly, at least one $\minor{A}{1..i-1,k}{1..i}$ for $k>i$ is
  nonzero, hence $a$ is a nonzero polynomial in $\Lambda_{i+1},\dots,\Lambda_n$ and $b,c$
  are nonzero polynomials in  $\Phi_j,\Psi_j$ for $j>i$, but constant in $\Phi_i$
  and $\Psi_i$. This is a contradiction, as the first column of the determinant,
  $ \begin{bmatrix}    b\\c  \end{bmatrix}$  can
  not be colinear with the second one. Hence $d_i\neq 0$.
  
  Therefore
  $
  \begin{bmatrix} e\\f  \end{bmatrix} = \frac{a}{d_i} \begin{bmatrix} b\\c\end{bmatrix} 
    $ which is
    \[
    \left\{
    \begin{array}{lcr}
      d_{i-1}Z_{i+1..n}^T X_{i+1..n}&=& \frac{1}{d_i}\sum_{j,k=i+1}^n \Lambda_k
      \left( d_i\minor{A}{1..i-1,k}{1..i-1,j}\right. \\
      &&-\left. \minor{A}{1..i-1,k}{1..i}\minor{A}{1..i}{1..i-1,j} \right)\Phi_j \\
      d_{i-1}Z_{i+1..n}^T Y_{i+1..n} &=& \frac{1}{d_i}\sum_{j,k=i+1}^n \Lambda_k
      \left(d_i\minor{A}{1..i-1,k}{1..i-1,j}\right. \\
      &&-\left. \minor{A}{1..i-1,k}{1..i}\minor{A}{1..i}{1..i-1,j}\right) \Psi_j\\
    \end{array}
    \right.
    \]
Applying variant \eqref{eq:dodgson:variant} of Lemma~\ref{lem:dodgson}  to $ \minor{A}{1..i, k}{1..i,j}$,
yields
    \[
    \left\{
    \begin{array}{lll}
      d_{i-1}Z_{i+1..n}^T X_{i+1..n}&=& \frac{1}{d_i}\sum_{j,k=i+1}^n \Lambda_k d_{i-1}\minor{A}{1..i,k}{1..i,j}\Phi_j\\
      d_{i-1}Z_{i+1..n}^T Y_{i+1..n}&=& \frac{1}{d_i}\sum_{ j,k=i+1}^n \Lambda_kd_{i-1}\minor{A}{1..i,k}{1..i,j} \Psi_j\\
    \end{array}
    \right.
    \]
and $H_{i+1}$ is verified.

We have proven that if $H_i$ is true, then either $H_{i+1}$ is also true or the
system~\eqref{eq:sys:rec} has a single solution and the Verifier randomly chose
precisely that $\lambda_i$. 
Therefore, suppose that $A$ has not generic rank profile, it means that some
$d_j=0$ and $H_j$ is false. But the Verifier checks that $H_1$ is true. 
If this is the case, then at least once, did the Verifier choose the value
expected by the dishonest Prover. This happens with probability lower than
$1/|S|$.

Finally, for the complexity, the Prover needs one Gaussian elimination to
compute $LU$ in time $O(n^\omega)$, then her extra work is just three triangular
solve in $O(n^2)$. The extra communication is three vectors, $\phi,\psi,\lambda$, and the
Verifier's work is four dot-products and one multiplication by the initial
matrix~$A$. 
  \end{proof}
\subsection{LDUP decomposition}\label{ssec:pldu}
With Protocol~\ref{cert:GRP}, when the matrix $A$ does not have
generic rank profile, any attempt to prove 
that it has generic rank profile will be detected w.h.p. (soundness).
However when it is the case, the verification will accept many possible vectors
$x,y,z$: any scaling of $z_i$ by $\alpha_i$ and $x_i,y_i$ by
$1/\alpha_i$ would be equally accepted for any non zero constants $\alpha_i$.
This slack correspond to our lack of specification of the diagonals'
shape in the used LU decomposition. 
Indeed, for any diagonal matrix with non zero elements, $LD\times
D^{-1}U$ is also a valid LU decomposition and yields $x,y$ and $z$
scaled as above. Specifying these diagonals is not necessary to prove
generic rank profileness, so we left it as is for this task.

However, for the determinant or the rank profile matrix certificates of
Sections~\ref{sec:det} and~\ref{sec:interractive:RPM}, we will need to
ensure that this scaling is independent from the choice of the vectors
$\phi,\psi,\lambda$. Hence we propose an updated protocol, where $L$ has to be
unit diagonal, and the prover has to first commit the 
main diagonal $D$ of $U$.

For an $n\times n$ triangular matrix $T$, its strictly 
triangular part is denoted $\widetilde{T}\in\F^{(n-1){\times}(n-1)}$: for
instance if $T$ is upper triangular, then $\widetilde{t}_{i,j}=t_{i,j+1}$ for
$j\geq i$ and $0$ otherwise. 


For $U$ an invertible upper triangular matrix we have for its diagonal
$(d_1,\ldots,d_n)$ and the associated diagonal matrix $D$,
that $U_1=D^{-1}U$ is unitary.
Thus, for any $\F^n\ni{}\psi=[\psi_1,\widetilde{\psi}]^T$:
$U\psi=DU_1\psi=D\left(\psi+\begin{smatrix}\widetilde{U_1}\widetilde{\psi}\\0\end{smatrix}\right)$.

So the idea is that the Prover will commit $D$ beforehand,
and that within a generic rank profile certificate, the Verifier will only
communicate 
$\widetilde{\phi}, \widetilde{\psi}$ and $\widetilde{\lambda}$ to obtain
$\widetilde{z}=\widetilde{\lambda}^T\widetilde{L}$, $\widetilde{x}=\widetilde{U}_1\widetilde{\phi}$ and $\widetilde{y}=\widetilde{U}_1\widetilde{\psi}$.
Then the Verifier will compute by herself the complete vectors.
This ensures that $L$ is unitary and that $U=DU_1$ with $U_1$ unitary.

Finally, if an invertible matrix does not have generic rank profile, we note
that it is also possible to incorporate the permutations, by
committing them in the beginning and reapplying them to the matrix
during the checks. 
The full certificate is given in Figure~\ref{cert:pldu}.

\newlength{\arrowlength}
\settowidth{\arrowlength}{\scriptsize{$\widetilde{x}_{i-1},\widetilde{y}_{i-1}$}}
\newcommand{\fxrightarrow}[1]{\xrightarrow{\mathmakebox[\arrowlength]{#1}}}
\newcommand{\fxleftarrow}[1]{\xleftarrow{\mathmakebox[\arrowlength]{#1}}}
\begin{figure}[htbp]
  \centering
  {\small
  \begin{tabular}{|l c l|}
    \hline\rule{0pt}{12pt}   
     Prover & & \hspace{15pt}Verifier \\
      \multicolumn{3}{|c|}{\hspace{-25pt}$A \in \mathbb{F}^{n{\times}n}$ non-singular} \\
    \hdashline\rule{0pt}{12pt}    
    $A=LDUP$ & $\fxrightarrow{P,D}$ & $P\checks{\in}\mathcal{S}_n$, $D\checks{\in}\mathcal{D}_n(\F)$\\
    \hdashline\rule{0pt}{12pt}   
    &&Choose $S\subset \F$\\
    &$\vdots$&\textbf{for} $i$ from $n$ downto $2$:\\
    $\begin{bmatrix}\widetilde{x}&\widetilde{y}\end{bmatrix}\leftarrow\widetilde{U}_1\begin{bmatrix} \widetilde{\phi}&\widetilde{\psi}\end{bmatrix}$& $\fxleftarrow{\phi_i,\psi_i}$ & $\phi_i,\psi_i\sample{S^2}$ \\
     & $\fxrightarrow{x_{i-1},y_{i-1}}$&  \\ 
    $\widetilde{z}\leftarrow\widetilde{\lambda}^T\widetilde{L}$ & $\fxleftarrow{\lambda_i}$ & $\lambda_i\sample{S}$  \\
    & $\fxrightarrow{z_{i-1}}$ &  \\
   & $\vdots$& \\
    \hdashline\rule{0pt}{12pt}    
    && $\phi_1,\psi_1,\lambda_1\sample{S^3}$ \\ 
    && $\begin{bmatrix}x&y\end{bmatrix}\leftarrow\begin{bmatrix}\phi&\psi\end{bmatrix}+\begin{bmatrix}\widetilde{x}&\widetilde{y}\\0&0\end{bmatrix}$ \\
    && $z^T\leftarrow\left(\lambda^T+\begin{bmatrix}\widetilde{z}^T&0\end{bmatrix}\right)$ \\
    && \hspace{-15pt}$z^TD \begin{bmatrix}x&y\end{bmatrix} \checks{=} (\lambda^T A)P^T \begin{bmatrix}\phi&\psi\end{bmatrix}$ \\
    \hline
  \end{tabular}
  \caption{\mbox{LDUP} decomposition (linear communication)}
  \label{cert:pldu}
  }
\end{figure}

\begin{theorem}\label{thm:pldu}
The Protocol of Figure~\ref{cert:pldu},
committing a permutation matrix $P$ and a diagonal matrix $D$ for an
invertible matrix $A$, such that there exists unitary triangular
matrices $L$ and $U$ with $A=LDUP$, is sound, with probability larger than
$1-\frac{1}{|S|}$, and perfectly complete.
For an $n{\times}n$ matrix, it requires less than $8n$ extra communications
and the computational cost for the Verifier is bounded by $\mu(A)+12n+o(n)$.
\end{theorem}
\begin{proof}
If the Prover is honest, then $A=LUP=LDU_1P$, so that for any choice of $\lambda$
and $\psi$ we have:
$\lambda^TAP^{T}\psi=\lambda^T LDU_1 \psi$, that is
$\begin{bmatrix}\widetilde{\lambda}^T&\lambda_n\end{bmatrix}
\left(I+\begin{smatrix}0&0\\\widetilde{L}&0\end{smatrix}\right)D\left(\begin{smatrix}0&\widetilde{U}\\0&0\end{smatrix}+I\right)\begin{bmatrix}\widetilde{\psi}\\\psi_n\end{bmatrix}
=z^Ty$
and the same is true for $\lambda$ and $\phi$, so that the protocol is
perfectly complete.

Now, the last part of the Protocol of Figure~\ref{cert:pldu} is
actually a verification that $AP^T$ has generic rank
profile, in other words that there exists lower and upper triangular
matrices $L^*$ and $U^*$ such that $AP^T=L^*U^*$. This
verification is sound by Theorem~\ref{th:GRP}. 
Next, the multiplication by the diagonal $D$ is performed by the
Verifier, so he is actually convinced that there exists lower and
upper triangular matrices $L^*$ and $U_1^*$ such that
$AP^T=L^* D U_1^*$.   
Finally, the construction of the vectors with the form
$a+\begin{smatrix}\widetilde{b}\\0\end{smatrix}$ is also done by the
Verifier, so he in fact has a guaranty that $L^*$ and $U_1^*$ are
unitary. 

Overall, if the Prover is dishonest, the Verifier will catch
him with the probability of Theorem~\ref{th:GRP}. 

Finally, for the complexity bounds, the extra communications are:
one permutation matrix $P$, a diagonal matrix $D$ and $6$ vectors 
$\widetilde{\lambda}$, $\widetilde{\phi}$, $\widetilde{\psi}$ and 
$\widetilde{z}$, $\widetilde{x}$ and $\widetilde{y}$. 
That is $n$ non-negative integers lower
than $n$ and $6(n-1)+n$ field elements. The arithmetic computations of the
Verifier are one multiplication by a diagonal matrix, $3$ vector sums, $4$
dot-products and one matrix-vector multiplication by $A$ (for
$(\lambda^TA)$), that is $n+3(n-1)+4(2n-1)$. 
\end{proof}
We, furthermore, have some guaranties on the actual values of $x,y,z$:
\begin{proposition}\label{cor:uniqueXY}
  Let $S$ be a finite subset of $\F$ in Protocol~\ref{cert:pldu}, if
  $\begin{bmatrix}x&y\end{bmatrix}\neq{}U_1\begin{bmatrix}\phi&\psi\end{bmatrix}$ 
  then the verification will pass with probability at most $\frac{2}{|S|}$. 
\end{proposition}
\begin{proof}
  Equation~\eqref{eq:det1} implies that, if the verification check
  pas\-ses, with $(z,x,y)$, then the vector
  $\begin{bmatrix}x_i&y_i\end{bmatrix}^T$ must be co-linear with the
  right column of this determinant, that can be written in the form
  $\begin{bmatrix} d_i \phi_i+b &d_i \psi_i +c \end{bmatrix}^T$ with
  $d_i\neq 0$ and $b$ and $c$ depending only on
  $\phi_k,\psi_k,x_k,y_k,\lambda_k,z_k$ with $k>i$.  
  Hence, any value $ \widetilde{x}_i,\widetilde{y}_i$, supplied by the
  Prover, must satisfy 
  \begin{equation}\label{eq:xiyi}
    \left|\begin{matrix}
        \phi_i+\widetilde{x}_i & d_i\phi_i + b\\ 
        \psi_i+\widetilde{y}_i & d_i\psi_i + c
      \end{matrix}\right|=0,
  \end{equation}
  when $\phi_i$ and $\psi_i$ are still unknown. 
  This condition is ensured for any $\phi_i$ and $\psi_i$ if and only if  
  $\begin{bmatrix}\widetilde{x}_i&\widetilde{y}_i\end{bmatrix} =
  \frac{1}{d_i} \begin{bmatrix}   b&c  \end{bmatrix}$.
  If the Prover is dishonest and if
  $\begin{bmatrix}x&y\end{bmatrix}\neq{}U_1\begin{bmatrix}\phi&\psi\end{bmatrix}$ then
  at least one couple $(\widetilde{x}_i,\widetilde{y}_i)$ is incorrect. Then,
  either the Verifier has chosen a couple of values $(\phi_i,\psi_i)$ making the 
  degree~$1$ determinant~\eqref{eq:xiyi} vanish, this happens with
  probability at most $1/|S|$, or System~\eqref{eq:sys:rec} has a unique
  solution $(z_i,\lambda_i)$. But if the latter is true and the final check succeeds
  then, as for Theorem~\ref{th:GRP}, at least once the Prover chose to have
  $1/|S|$ chances that the Verifier picked the unique possibility for $\lambda_j$,
  $i\geq{}j\geq{}1$.
  Overall, the Verification thus fails with probability at most $1-\frac{2}{|S|}$.
\end{proof}

\begin{remark}\label{rem:zLw}
Correctness of the vector $z$ can also be ensured with the same probability: 
for the singular System~\eqref{eq:sys:rec}, with respect to the
unknowns $\Lambda_i$ and $Z_i$, to have rank at least one, it is sufficient that one of
$X_i$ or $Y_i$ is non zero. The Verifier, knowing $\widetilde{x_i}$, can ensure
this by restricting the set of choices for
$\phi_i\in{}S{\setminus}\{-\widetilde{x_i}\}$. Thus if $x_i$ and $y_i$ are correct, the
Prover will have to provide a correct associated $z_i$ or increase the
probability of being caught.
\end{remark}
\section{Linear communication interactive certificates}
\label{sec:interractive}

In this section, we give linear space communication certificates for
the determinant,
the column/row rank profile of a matrix, and for the rank profile matrix.

\subsection{Linear communication certificate for the
  determinant}\label{sec:det}
Existing certificates for the determinant are either optimal for the Prover in
the dense case, using the strategy of \cite[Theorem~5]{kns11} over a PLUQ
decomposition, but quadratic in communication; or linear in communication, using
\cite[Theorem~14]{jgd:2016:gammadet}, but using a reduction to the
characteristic polynomial. 
In the sparse case the determinant and the characteristic polynomial both
reduce to the same minimal polynomial computations and therefore the latter
certificate is currently optimal for the Prover. Now in the dense case, while
the determinant and characteristic polynomial both reduce to matrix
multiplication, the determinant, via a single PLUQ decomposition is more
efficient in practice~\cite{Pernet:2007:charp}. 
Therefore, we propose here an alternative in
the dense case: use only one PLUQ decomposition for the Prover while keeping
linear extra communications and $O(n)+\mu(A)$ operations for the Verifier.
The idea is to extract the 
information of a LDUP decomposition without communicating it: 
one uses Protocol~\ref{cert:pldu} for $A=LDUP$ with $L$ and $U$
unitary, but kept on the Prover side, and then the Verifier only has to compute
$Det(A)=Det(D)Det(P)$, with $n-1$ additional field operations.

\begin{corollary}\label{thm:det}
For an $n{\times}n$ matrix, there exists a sound and perfectly complete protocol
for the determinant over a field using less than $8n$ extra
communications and with computational cost for the Verifier bounded by
$\mu(A)+13n+o(n)$. 
\end{corollary}
As a comparison, the protocol of \cite[Theorem~14]{jgd:2016:gammadet} reduces
to {\sc{CharPoly}} instead of PLUQ for the Prover, requires $5n$ extra
communications and $\mu(A)+13n+o(n)$ operations for the Verifier as well. 
Also the new protocol requires $3n$ random field elements for any field, where that of \cite[Theorem~14]{jgd:2016:gammadet} requires $3$
random elements but a field larger than $n^2$.

For instance, using the routines shown in Table~\ref{tab:io},
the determinant of an 
$50k{\times}50k$ random dense
matrix can be computed in about
24~minutes,
where with the certificate of Figure~\ref{cert:pldu}, the overhead of the Prover is less than 
5s 
and the Verifier time is about
1s.



\subsection{Column or row rank profile certificate}
\label{sec:interractive:CRP}
In Figure~\ref{cert:upper_rank} and~\ref{cert:lower_rank}, we first
recall the two linear time and space certificates for an upper and a
lower bound to the rank that constitute a rank certificate. We present
here the variant sketched in~\cite[\S~2]{eberly15} of the certificates
of~\cite{dk14}.  An upper bound $r$ on the rank is certified by the capacity for the
Prover to generate any vector sampled from the image of $A$ by a linear
combination of $r$ column of $A$. A lower bound $r$  is certified by the capacity
for the Prover to recover the unique coefficients of a linear combination of $r$
linearly independent columns of $A$.

\begin{figure}[htbp]
    \centering
    \begin{tabular}{|l c l|}
        \hline
        Prover &  & Verifier \\
        & $ A \in \mathbb{F}^{m \times n} $ &  \\
        \hdashline\rule{0pt}{12pt}
        $r$ s.t. $ \rank(A) \leq r $ & $\xrightarrow{\text{$r$}} $ & \\
        \hdashline\rule{0pt}{12pt}
        & & Choose $S \subset \mathbb{F}$ \\
        & $\xleftarrow{\text{$w$}}$ & $ v \sample{S^n}, w = Av $ \\
        \hdashline\rule{0pt}{12pt}
        $A\gamma = w$ & $\xrightarrow{\text{$\gamma$}}$ & $|\gamma|_H\stackrel{?}{=}r$\\
        & & $A\gamma \stackrel{?}{=} w $\\
        \hline
    \end{tabular}
    \caption{Upper bound on the rank of a matrix}
    \label{cert:upper_rank}
\vspace{\spaceafterprotocol}\end{figure}

\begin{theorem}
    \label{th:upper_rank}
    Let $A \in \mathbb{F}^{m \times n}$, and let $S$ be a finite subset
    of $\mathbb{F}$.
    The interactive certificate~\ref{cert:upper_rank} of
    an upper bound for the rank of $A$ is sound, with probability larger than
    $1-\frac{1}{|S|}$, perfectly complete, occupies $2n$ communication space,
    can be computed in $LINSYS(r)$ and verified in $2\mu(A) + n$ time.
\end{theorem}

%
%
%

\begin{figure}[htbp]
    \centering
    \begin{tabular}{|p{94pt} c l|}
        \hline
        \multicolumn{1}{|l}{Prover} &  & Verifier \\
        & $A\in\mathbb{F}^{m{\times}n}$ &  \\
        \hdashline\rule{0pt}{12pt}
        \hspace{-3pt}$c_1,.., c_r$ indep. cols of A &
        $\xrightarrow{\text{$c_1, .., c_r$}} $ & \\
        \hdashline\rule{0pt}{12pt}
        & & Choose $S \subset \mathbb{F}$ \\
        & $\xleftarrow{v}$ &  $\alpha = \begin{dcases*}
          \alpha_{c_j} \sample{S^*}  \\ 0\text{ otherwise} \end{dcases*} $ \\
        & & $v = A\alpha$\\
        \hdashline\rule{0pt}{12pt}
        Solve $A\beta = v$ & $\xrightarrow{\text{$\beta$}}$ & $\beta\checks{=}\alpha$\\
        \hline
    \end{tabular}
    \caption{Lower bound on the rank of a matrix}
    \label{cert:lower_rank}
\vspace{\spaceafterprotocol}\end{figure}

\begin{theorem}
    \label{th:lower_rank}
    Let $A \in \mathbb{F}^{m \times n}$, and let $S$ be a finite subset
    of $\mathbb{F}$.
    The interactive certificate~\ref{cert:lower_rank} of a lower bound
    for the rank of $A$ is sound, , with probability larger than
    $1-\frac{1}{|S|}$, perfectly complete and occupies $n + 2r$ communication
    space, can be computed in $LINSYS(r)$ and verified in $\mu(A)+r$ operations.
\end{theorem}
We now consider a column rank profile certificate: the Prover is given a matrix $A$, and answers
the column rank profile of $A$, $\mathcal{J} = (c_1, \dots, c_r)$. 
In order to certify this column rank profile, we need to certify two properties: 
\begin{compactenum}
\item the columns given by $\mathcal{J}$ are linearly independent; \label{CRP:step:indep}
\item the columns given by $\mathcal{J}$ form the lexicographically smallest set \label{CRP:step:minimal}
  of independent columns of $A$.
\end{compactenum}

Property~\ref{CRP:step:indep} is verified by Certificate~\ref{cert:lower_rank},
as 
it checks wether a set of columns are indeed linearly independent.
Property~\ref{CRP:step:minimal} could be certified by successive applications of Certificate~\ref{cert:upper_rank}: at step $i$, checking that the rank of $A_{*, (0, \dots, c_i - 1)}$
is at most $i-1$ would certify that there is no column located between $c_{i-1}$ and $c_i$ in $A$ 
which increases the rank of $A$. Hence, it would prove  the minimality of $\mathcal{J}$.
However, this method requires $O(nr)$  communication space.

Instead, we reduce these communication by seeding all challenges from
a single $n$ dimensional vector, and by compressing the responses with a random
projection.
The right triangular equivalence certificate  plays here a central role,
ensuring the lexicographic minimality of $\mathcal{S}$.
More precisely, the Verifier chooses a vector $v \in \F^{n}$ uniformly at random
and sends it to the Prover.
Then, for each index $c_k\in\mathcal{S}$ the Prover computes the linear
combination of the first $c_k-1$  columns of $A$ using the first $c_k-1$
coefficients of $v$ and has to prove that it can be generated from the $k-1$
columns $c_1,\dots,c_{k-1}$. This means, find  a vector $\gamma^{(k)}$ solution to the system:
$$
\begin{bmatrix}
  A_{*,c_1} &   A_{*,c_2} & \dots &   A_{*,c_{k-1}}  
\end{bmatrix}
\gamma^{(k)}
=
A 
\begin{smatrix}
  v_1\\\vdots\\ v_{c_k-1}\\0\\\vdots
\end{smatrix}.
$$

Equivalently, find a strictly upper triangular matrix $\Gamma$ such that:
$$
\begin{bmatrix}
  A_{*,c_1} &   A_{*,c_2} & \dots &   A_{*,c_{r-1}}  
\end{bmatrix}
\Gamma =
A
\underbrace{\begin{smatrix}
  v_1&v_1 & \cdots & \cdots&v_1\\
  \vdots&\vdots&  \vdots& \vdots&\vdots\\
  v_{c_1-1}&\vdots & \vdots&\vdots&\vdots\\
  0 & v_{c_2-1}&\vdots& \vdots&\vdots\\
  0 & 0 & \ddots & \vdots&\vdots\\
  0 & 0 & 0 & v_{c_r-1}&\vdots\\
  0 & 0 & 0 & 0&v_{n}
\end{smatrix}}_{V}.
$$

Note that $V=\text{Diag}(v_1,\dots,v_n)W$ where $W = [\mathds{1}_{i<c_{j+1}}]_{i,j}$ (with $c_{r+1}=n+1$ by convention)
In order to avoid having to transmit the whole $r\times r$ upper triangular matrix
$\Gamma$, the Verifier only checks a random projection $x$ of it, using the
triangular equivalence Certificate~\ref{cert:triangular_equivalence}.
We then propose the certificate in Figure~\ref{cert:crp}.
\begin{figure}[htbp]
  \centering
  {\def\arraystretch{1.2}
    \begin{tabular}{|p{2.7cm} c l|}
        \hline
        Prover & & Verifier \\
        & $A \in \mathbb{F}^{m\times n}$ & \\
        \hdashline\rule{0pt}{12pt}
        $(c_1,.., c_r)$ CRP of $A$ & $\xrightarrow{\text{$(c_1,.., c_r)$}}$ &
        $\rank{A} \checks{\geq} r$  by Cert.~\ref{cert:lower_rank}\\
        \hdashline
        & & Choose $S \subset \mathbb{F}$ \\
        & $\xleftarrow{v}$ & $v \sample{S^{n}}$\\
        $V=\text{Diag}(v_i)W$
        &&    $W = [\mathds{1}_{i<c_{j+1}} ]
        $\\

        $\Gamma$ upper tri.  s.t. $A_{*,\{c_1,..,c_r\}}\Gamma = AV$ & &$D\leftarrow \text{Diag}(v_i)$\\
        \hdashline
        $y = \Gamma x$ 
        & $ \xleftrightarrow{x \text{  (Cert.~\ref{cert:triangular_equivalence})  } y}$& $x \sample{S^{r}}$ \\
        \hdashline
 && $z \leftarrow D(Wx)$\\
 && $z_{c_j}\leftarrow{}z_{c_j}{-}y_j, j=1..r$\\
 && $Az\checks{=}0$\\
        \hline
    \end{tabular}
    \caption{Certificate for the column rank profile}
    \label{cert:crp}
    }
\vspace{\spaceafterprotocol}\end{figure}
\begin{theorem}
    For $A \in \F^{m \times n}$ and $S \subset \mathbb{F}$,
    certificate~\ref{cert:crp} is sound, with probability larger than
    $1-\frac{1}{|S|}$, 
    perfectly complete, with a Prover
    computational cost bounded by $O(mnr^{\omega - 2})$,
    a communication space complexity bounded by $2n+4r$
    and a Verifier cost bounded by $2\mu(A)+n+3r$.
\end{theorem}

\begin{proof}
    If the Prover is honest, the protocol corresponds first to an application of
    Theorem~\ref{th:lower_rank} to certify that $\mathcal{J}$ is a set of
    independent columns. This certificate is perfectly complete.
    Second the protocol also uses challenges from
    Certificate~\ref{cert:upper_rank}, 
    which is perfectly complete, together with
    Certificate~\ref{cert:triangular_equivalence}, which is perfectly complete
    as well. The latter certificate is used on $A_{*, \mathcal{J}}$, a
    regular submatrix, as $\mathcal{J}$ is a set of independent columns of $A$.
    The final check then corresponds to
    $A(D(Wx)) -A_{*,\{c_1,..c_r\}}y \stackrel{?}{=} 0$ and, overall,
    Certificate~\ref{cert:crp} is perfectly complete. 

    If the Prover is dishonest, then either the set of  columns in $\mathcal{J}$ are  not
    linearly 
    independent, which will be caught by the Verifier with probability at least 
    $1 - \frac{1}{|S|}$, from Theorem~\ref{th:lower_rank},
    or
    $\mathcal{J}$ is not lexicographically minimal, or the rank of $A$ is not $r$.
    If the rank is wrong, it will not be possible for the prover to find a
    suitable $\Gamma$. This will be caught by the verifier with probability 
    $1-\frac{1}{|S|}$, from 
    Theorem~\ref{th:triangular_equivalence}.
    Finally, if $\mathcal{J}$ 
    is not lexicographically minimal, there exists at least one column 
    $c_k \notin    \mathcal{J}, c_i < c_k < c_{i+1}$ for some fixed 
    $i$ such that $\{c_1, \dots, c_i\} \cup \{c_k\}$ form a set of linearly
    independant columns of $A$.  
    This means that $\rank(A_{*, 1, \dots, c_{i+1} - 1}) = i+1$, whereas it was
    expected to be $i$. 
    Thus, the prover cannot reconstruct a suitable triangular $\Gamma$ and this
    will be detected by the verifier also with probability $1-\frac{1}{|S|}$, as
    shown in Theorem~\ref{th:triangular_equivalence}).

    The Prover's time complexity is that of computing a $PLUQ$ decomposition of
    $A$.    
    The transmission of $v, x$ and $y$ yields a 
    communication space of $n+2r$. Finally, in addition to
    Protocol~\ref{cert:lower_rank}, the Verifier computes $Wx$ as a prefix sum with $r-1$
    additions, multiplies it by $D$, then substracts $y_i$ at the $r$
    correct positions and finally multiplies by $A$ for a total cost bounded by
    $2\mu(A)+n+3r-1$.
\end{proof}


\subsection{Rank profile matrix certificate}
\label{sec:interractive:RPM}

We propose an interactive certificate for the rank profile matrix based
on~\cite[Algorithm~4]{dps16}: first computing the row and column support of the
rank profile matrix, using Certificate~\ref{cert:crp} twice for the row and
column rank profiles, then computing the rank profile matrix of the invertible
submatrix of $A$ lying on this grid.

In the following we then only focus on a certificate for the rank profile
matrix of an invertible matrix. It relies on an LUP decomposition that reveals
the rank profile matrix. From Theorem~\ref{lem:echelonized}, this is the case if
and only if $P^TUP$ is upper triangular. 
Protocol~\ref{cert:RPM:linearcomm} thus gives an interactive certificate that combines
Certificate~\ref{cert:pldu} for a LDUP decomposition with a certificate that $P^TUP$ is upper
triangular. The latter is achieved by
Certificate~\ref{cert:triangular_equivalence} showing that $P^T$ and $P^TU$ are
left upper triangular equivalent, but since $U$ is unknown to the Verifier, the
verification is done on a random right projection with the vector $\phi$  used in Certificate~\ref{cert:pldu}.
  \begin{figure}[htbp]
    {
      \centering
  {\def\arraystretch{1.2}
    \begin{tabular}{|p{2.2cm} c l|}
        \hline
        Prover & & Verifier \\
        \multicolumn{3}{|c|}{$A \in \mathbb{F}^{n\times n}$ invertible} \\
        \hdashline
          $A=LDUP$, with $P=\RPM{A}$
        &
        $\xrightarrow{\ P,D \ }$ &$P \checks{\in}\mathcal{S}_n$, $D\checks{\in}\mathcal{D}_n(\F)$\\
        \hline
        \multicolumn{3}{|c|}{Protocol~\ref{cert:triangular_equivalence} :
          {\small $P^T$ and $P^TU$ are left up. tri. equiv. with random proj.}}\\
        \hdashline
        $\overline{U}\leftarrow P^TUP$ & $\xrightarrow{\overline{U} \text{ is upper tri.}}$&\\
        \hdashline
        &&Choose $S\subset \F$\\ 
        &$\xleftarrow{e_i}$&\textbf{for } $i=1..n$, $e_i\sample{S}$\\
        $f^T\leftarrow e^T\overline{U}$ &  $\xrightarrow{f_i}$& \\
        &  
        $\begin{array}{|c|}
          \hline
          \text{\scriptsize Protocol~\ref{cert:pldu} on A}\\
               \hdashline
               \xleftarrow{ \begin{smatrix} \widetilde{\phi}&\widetilde{\psi} \end{smatrix}
               } \\
               \xrightarrow{\begin{smatrix}\widetilde{x}&\widetilde{y}\end{smatrix}} \\
               \hline
        \end{array}$
        & $\phi,\psi\sample{S^n}$   \\
        && Now $
        \begin{bmatrix} x&y \end{bmatrix} =U \begin{bmatrix} \phi&\psi \end{bmatrix}
        $\\
        \hdashline
        && $ e^TP^T x
        \checks{=}  f^T P^T   \phi
        $  \\
        \hline
    \end{tabular}
    \caption{Rank profile matrix of an invertible matrix}
    \label{cert:RPM:linearcomm}
  }
  }
\vspace{\spaceafterprotocol}\end{figure}

  \begin{theorem}
    Protocol~\ref{cert:RPM:linearcomm} is sound, with probability greater than
    $1-\frac{2}{|S|}$, and perfectly complete. The Prover cost is $\GO{n^\omega}$
    field operations, the communication space is bounded by $10n$ and the
    Verifier cost is bounded by $\mu(A)+16n$.
  \end{theorem}

\begin{proof}
  If the Prover is dishonest and  $\overline{U}=P^TUP$ is not upper triangular,
then let $(i,j)$ be the lexicographically minimal coordinates such that $i>j$ and
$\overline{U}_{i,j}\neq 0$.
Now either
$\begin{bmatrix} x&y  \end{bmatrix} \neq U \begin{bmatrix} \phi&\psi  \end{bmatrix}$,
and the verification will then fail to detect it with probability less than $\frac{2}{|S|}$, from
Proposition~\ref{cor:uniqueXY}.
Or one can write
$e^TP^T  x- f^TP^T \phi = (e^T\overline{U}-f^T) P\phi=0$.
  If
  \begin{equation}\label{eq:qtuq}
    e^TP^TUP-f^T=0.
  \end{equation}
  is not satisfied, then a random $\phi$ will fail to detect it
  with probability less than $\frac{1}{|S|}$, since
$e,\overline{U}$ and $f$ are set before  choosing for $\phi$.
 At the time of commiting $f_j$, the value of $e_i$
  is still unknown, hence $f_j$ is constant in the symbolic variable
  $E_i$. Thus the $j$-th coordinate in~\eqref{eq:qtuq} is a nonzero polynomial
  in $E_j$ and therefore vanishes with probability $1/|S|$ when sampling the
  values of $e$ uniformly. Hence, overall if $P^TUP$ is not upper triangular, the
  verification will fail to detect it with probability at most $2/|S|$.
 \end{proof}

Finally, the rank profile matrix of any matrix, even a singular one, can thus be
verified with two applications of Certificate~\ref{cert:crp} (one for the row
rank profile and one for the column rank profile, themselves calling
Certificate~\ref{cert:lower_rank} only once), followed by
Certificate~\ref{cert:RPM:linearcomm} on the $r{\times}r$ selection of
lexicographically minimal independent rows and columns.  
Overall this is $4\mu(A)+2n+21r$ operations for the Verifier, and $3n+16r$
communications. 

\bibliographystyle{abbrvurl}
\makeatletter
\def\thebibliography#1{%
\ifnum\addauflag=0\addauthorsection\global\addauflag=1\fi
     \section[References]{
        {References} 
          \vskip -9pt  
         \@mkboth{{\refname}}{{\refname}}%
     }%
     \list{[\arabic{enumi}]}{%
         \settowidth\labelwidth{[#1]}%
         \leftmargin\labelwidth
         \advance\leftmargin\labelsep
         \advance\leftmargin\bibindent
         \parsep=0pt\itemsep=1pt 
         \itemindent -\bibindent
         \listparindent \itemindent
         \usecounter{enumi}
     }%
     \let\newblock\@empty
     \raggedright 
     \sloppy
     \sfcode`\.=1000\relax
     \vspace*{0.5mm}
     \small
}
\makeatother

\bibliography{rpm_certificate}

\begin{thebibliography}{10}

\bibitem{ckl13}
H.~Y. Cheung, T.~C. Kwok, and L.~C. Lau.
\newblock Fast {Matrix} {Rank} {Algorithms} and {Applications}.
\newblock {\em Journal of the ACM}, 60(5):31:1--31:25, Oct. 2013.
\newblock \href {http://dx.doi.org/10.1145/2528404}
  {\path{doi:10.1145/2528404}}.

\bibitem{Costello:2015:gepetto}
C.~Costello, C.~Fournet, J.~Howell, M.~Kohlweiss, B.~Kreuter, M.~Naehrig,
  B.~Parno, and S.~Zahur.
\newblock Geppetto: Versatile verifiable computation.
\newblock In {\em 2015 {IEEE} Symposium on Security and Privacy, {SP} 2015, San
  Jose, CA, USA, May 17-21, 2015}, pages 253--270, 2015.
\newblock \href {http://dx.doi.org/10.1109/SP.2015.23}
  {\path{doi:10.1109/SP.2015.23}}.

\bibitem{Dod1866}
C.~L. Dodgson.
\newblock Condensation of {Determinants}, {Being} a {New} and {Brief} {Method}
  for {Computing} their {Arithmetical} {Values}.
\newblock {\em Proceedings of the Royal Society of London}, 15:150--155, 1866.
\newblock URL: \url{http://www.jstor.org/stable/112607}.

\bibitem{dk14}
J.-G. Dumas and E.~Kaltofen.
\newblock Essentially optimal interactive certificates in linear algebra.
\newblock In K.~Nabeshima, editor, {\em {ISSAC}'2014}, pages 146--153. ACM
  Press, New York, July 2014.
\newblock \href {http://dx.doi.org/10.1145/2608628.2608644}
  {\path{doi:10.1145/2608628.2608644}}.

\bibitem{jgd:2016:gammadet}
J.-G. Dumas, E.~Kaltofen, E.~Thom\'e, and G.~Villard.
\newblock Linear time interactive certificates for the minimal polynomial and
  the determinant of a sparse matrix.
\newblock In X.-S. Gao, editor, {\em {ISSAC}'2016}, pages 199--206. ACM Press,
  New York, July 2016.
\newblock \href {http://dx.doi.org/10.1145/2930889.2930908}
  {\path{doi:10.1145/2930889.2930908}}.

\bibitem{DPS:2013}
J.-G. Dumas, C.~Pernet, and Z.~Sultan.
\newblock Simultaneous computation of the row and column rank profiles.
\newblock In M.~Kauers, editor, {\em {ISSAC}'2013}, pages 181--188. ACM Press,
  New York, June 2013.
\newblock \href {http://dx.doi.org/10.1145/2465506.2465517}
  {\path{doi:10.1145/2465506.2465517}}.

\bibitem{dps15}
J.-G. Dumas, C.~Pernet, and Z.~Sultan.
\newblock Computing the rank profile matrix.
\newblock In Yokoyama \cite{2015:ISSAC:Yokoyama}, pages 149--156.
\newblock \href {http://dx.doi.org/10.1145/2755996.2756682}
  {\path{doi:10.1145/2755996.2756682}}.

\bibitem{dps16}
J.-G. Dumas, C.~Pernet, and Z.~Sultan.
\newblock Fast computation of the rank profile matrix and the generalized
  {Bruhat} decomposition.
\newblock {\em Journal of Symbolic Computation}, 2016.
\newblock in press.
\newblock \href {http://dx.doi.org/10.1016/j.jsc.2016.11.011}
  {\path{doi:10.1016/j.jsc.2016.11.011}}.

\bibitem{eberly15}
W.~Eberly.
\newblock A new interactive certificate for matrix rank.
\newblock Technical Report 2015-1078-11, University of Calgary, June 2015.
\newblock URL:
  \url{http://prism.ucalgary.ca/bitstream/1880/50543/1/2015-1078-11.pdf}.

\bibitem{Fiat:1986:Shamir}
A.~Fiat and A.~Shamir.
\newblock How to prove yourself: Practical solutions to identification and
  signature problems.
\newblock In A.~M. Odlyzko, editor, {\em Advances in Cryptology - {CRYPTO'86}},
  volume 263 of {\em \textit{LNCS}}, pages 186--194. Springer-Verlag, 1987,
  11--15~Aug. 1986.
\newblock URL: \url{http://www.cs.rit.edu/~jjk8346/FiatShamir.pdf}.

\bibitem{freivalds79}
R.~Freivalds.
\newblock Fast probabilistic algorithms.
\newblock {\em Mathematical Foundations of Computer Science, \textit{LNCS}},
  74:57--69, Sept. 1979.
\newblock \href {http://dx.doi.org/10.1007/3-540-09526-8_5}
  {\path{doi:10.1007/3-540-09526-8_5}}.

\bibitem{Goldwasser:2008:delegating}
S.~Goldwasser, Y.~T. Kalai, and G.~N. Rothblum.
\newblock Delegating computation: interactive proofs for muggles.
\newblock In C.~Dwork, editor, {\em {STOC}'2008}, pages 113--122. ACM Press,
  May 2008.
\newblock \href {http://dx.doi.org/10.1145/1374376.1374396}
  {\path{doi:10.1145/1374376.1374396}}.

\bibitem{jps13}
C.-P. Jeannerod, C.~Pernet, and A.~Storjohann.
\newblock Rank-profile revealing gaussian elimination and the {CUP} matrix
  decomposition.
\newblock {\em Journal of Symbolic Computation}, 56:pages 46--68, 2013.
\newblock \href {http://dx.doi.org/10.1016/j.jsc.2013.04.004}
  {\path{doi:10.1016/j.jsc.2013.04.004}}.

\bibitem{kns11}
E.~L. Kaltofen, M.~Nehring, and B.~D. Saunders.
\newblock Quadratic-time certificates in linear algebra.
\newblock In A.~Leykin, editor, {\em {ISSAC}'2011}, pages 171--176. ACM Press,
  New York, June 2011.
\newblock \href {http://dx.doi.org/10.1145/1993886.1993915}
  {\path{doi:10.1145/1993886.1993915}}.

\bibitem{Pernet:2007:charp}
C.~Pernet and A.~Storjohann.
\newblock Faster algorithms for the characteristic polynomial.
\newblock In C.~W. Brown, editor, {\em {ISSAC}'2007}, pages 307--314. ACM
  Press, New York, July 29 -- August 1 2007.
\newblock \href {http://dx.doi.org/10.1145/1277548.1277590}
  {\path{doi:10.1145/1277548.1277590}}.

\bibitem{sy15}
A.~Storjohann and S.~Yang.
\newblock A {Relaxed} {Algorithm} for {Online} {Matrix} {Inversion}.
\newblock In Yokoyama \cite{2015:ISSAC:Yokoyama}, pages 339--346.
\newblock \href {http://dx.doi.org/10.1145/2755996.2756672}
  {\path{doi:10.1145/2755996.2756672}}.

\bibitem{2015:ISSAC:Yokoyama}
K.~Yokoyama, editor.
\newblock {\em {ISSAC}'2015}. ACM Press, New York, July 2015.

\end{thebibliography}

\end{document}